\documentclass[twocolumn,aps,nofootinbib,superscriptaddress]{revtex4}

\usepackage{dsfont}
\usepackage{amsmath}
\usepackage{amssymb}
\usepackage{amsthm}
\usepackage{graphicx}

\newcommand{\comment}[1]{}

\newcommand{\abs}[1]{\ensuremath{|#1|}}

\newcommand{\Abs}[1]{\ensuremath{\left|#1\right|}}

\newcommand{\norm}[2]{\ensuremath{|\!|#1|\!|_{#2}}}
\newcommand{\normbig}[2]{\ensuremath{\big|\!\big|#1\big|\!\big|_{#2}}}

\newcommand{\tr}{\textnormal{tr}}
\newcommand{\trace}[1]{\ensuremath{\tr (#1)}}
\newcommand{\Trace}[1]{\ensuremath{\tr \left( #1 \right)}}

\newcommand{\ptr}[1]{\textnormal{tr}_{\textnormal{\tiny #1}}}

\newcommand{\supp}[1]{\textnormal{supp}\, \{ #1 \}}
\newcommand{\rank}{\textnormal{rank}}

\newcommand{\idx}[2]{{#1}_{\textnormal{\tiny #2}}}

\newcommand{\ket}[1]{| #1 \rangle}
\newcommand{\keti}[2]{| #1 \rangle_{\textnormal{\tiny #2}}}

\newcommand{\braket}[2]{\langle #1 | #2 \rangle}

\newcommand{\bracket}[3]{\langle #1 | #2 | #3 \rangle}

\newcommand{\proj}[2]{| #1 \rangle\!\langle #2 |}
\newcommand{\proji}[3]{| #1 \rangle\!\langle #2 |_{\textnormal{\tiny #3}}}

\newcommand{\iso}{\cong}

\newcommand{\kron}{\otimes}
\newcommand{\eps}{\varepsilon}

\newcommand{\h}{\ensuremath{\mathcal{H}}}
\newcommand{\hi}[1]{\ensuremath{\mathcal{H}_{\textnormal{\tiny #1}}}}
\newcommand{\hA}{\hi{A}}
\newcommand{\hB}{\hi{B}}
\newcommand{\hC}{\hi{C}}
\newcommand{\hD}{\hi{D}}

\newcommand{\hX}{\hi{X}}
\newcommand{\hY}{\hi{Y}}
\newcommand{\hR}{\hi{R}}
\newcommand{\hAB}{\hi{AB}}
\newcommand{\hAC}{\hi{AC}}

\newcommand{\hABC}{\hi{ABC}}

\newcommand{\idi}[1]{\ensuremath{\mathds{1}_{\textnormal{\tiny #1}}}}

\newcommand{\idA}{\idi{A}}
\newcommand{\idB}{\idi{B}}
\newcommand{\idC}{\idi{C}}

\newcommand{\idE}{\idi{E}}
\newcommand{\idX}{\idi{X}}

\newcommand{\idAB}{\idi{AB}}

\newcommand{\idBC}{\idi{BC}}

\newcommand{\opidi}[1]{\ensuremath{\mathcal{I}_{\textnormal{\tiny #1}}}}
\newcommand{\opidA}{\opidi{A}}

\newcommand{\linops}[1]{\ensuremath{\mathcal{L}(#1)}}
\newcommand{\hermops}[1]{\ensuremath{\mathcal{L}^\dagger(#1)}}
\newcommand{\posops}[1]{\ensuremath{\mathcal{P}(#1)}}

\newcommand{\normstates}[1]{\ensuremath{\mathcal{S}_{=}(#1)}}
\newcommand{\subnormstates}[1]{\ensuremath{\mathcal{S}_{\leq}(#1)}}

\newcommand{\rhot}{\ensuremath{\tilde{\rho}}}
\newcommand{\rhob}{\ensuremath{\bar{\rho}}}
\newcommand{\rhoh}{\ensuremath{\hat{\rho}}}
\newcommand{\rhoA}{\ensuremath{\idx{\rho}{A}}}
\newcommand{\rhoB}{\ensuremath{\idx{\rho}{B}}}

\newcommand{\rhoAB}{\ensuremath{\idx{\rho}{AB}}}
\newcommand{\rhotAB}{\ensuremath{\idx{\rhot}{AB}}}
\newcommand{\rhobAB}{\ensuremath{\idx{\rhob}{AB}}}

\newcommand{\rhoAC}{\ensuremath{\idx{\rho}{AC}}}
\newcommand{\rhotAC}{\ensuremath{\idx{\rhot}{AC}}}
\newcommand{\rhobAC}{\ensuremath{\idx{\rhob}{AC}}}

\newcommand{\rhoABC}{\ensuremath{\idx{\rho}{ABC}}}
\newcommand{\rhoABE}{\ensuremath{\idx{\rho}{ABE}}}
\newcommand{\rhotABC}{\ensuremath{\idx{\rhot}{ABC}}}
\newcommand{\rhobABC}{\ensuremath{\idx{\rhob}{ABC}}}

\newcommand{\sigmat}{\ensuremath{\tilde{\sigma}}}
\newcommand{\sigmab}{\ensuremath{\bar{\sigma}}}
\newcommand{\sigmah}{\ensuremath{\hat{\sigma}}}

\newcommand{\sigmaB}{\ensuremath{\idx{\sigma}{B}}}
\newcommand{\sigmaC}{\ensuremath{\idx{\sigma}{C}}}

\newcommand{\sigmatB}{\ensuremath{\idx{\sigmat}{B}}}
\newcommand{\sigmabB}{\ensuremath{\idx{\sigmab}{B}}}

\newcommand{\taut}{\ensuremath{\tilde{\tau}}}
\newcommand{\taub}{\ensuremath{\bar{\tau}}}
\newcommand{\tauh}{\ensuremath{\hat{\tau}}}

\newcommand{\tauB}{\ensuremath{\idx{\tau}{B}}}

\newcommand{\tauAB}{\ensuremath{\idx{\tau}{AB}}}
\newcommand{\tauAC}{\ensuremath{\idx{\tau}{AC}}}
\newcommand{\tauABC}{\ensuremath{\idx{\tau}{ABC}}}
\newcommand{\tauAD}{\ensuremath{\idx{\tau}{AD}}}
\newcommand{\tauCD}{\ensuremath{\idx{\tau}{CD}}}
\newcommand{\tauCDE}{\ensuremath{\idx{\tau}{CDE}}}
\newcommand{\tautAB}{\ensuremath{\idx{\taut}{AB}}}

\newcommand{\taubCD}{\ensuremath{\idx{\taub}{CD}}}
\newcommand{\tautCD}{\ensuremath{\idx{\taut}{CD}}}

\newcommand{\omegat}{\ensuremath{\tilde{\omega}}}
\newcommand{\omegab}{\ensuremath{\bar{\omega}}}
\newcommand{\omegah}{\ensuremath{\hat{\omega}}}
\newcommand{\omegabD}{\ensuremath{\idx{\omegab}{D}}}
\newcommand{\omegatD}{\ensuremath{\idx{\omegat}{D}}}

\newcommand{\phiA}{\ensuremath{\idx{\phi}{A}}}
\newcommand{\phiB}{\ensuremath{\idx{\phi}{B}}}
\newcommand{\phiAB}{\ensuremath{\idx{\phi}{AB}}}

\newcommand{\psiB}{\ensuremath{\idx{\psi}{B}}}
\newcommand{\psiAB}{\ensuremath{\idx{\psi}{AB}}}
\newcommand{\phit}{\ensuremath{\tilde{\phi}}}

\newcommand{\PiA}{\ensuremath{\idx{\Pi}{A}}}
\newcommand{\PiB}{\ensuremath{\idx{\Pi}{B}}}

\newcommand{\PiU}{\ensuremath{\idx{\Pi}{U}}}
\newcommand{\PiV}{\ensuremath{\idx{\Pi}{V}}}
\newcommand{\PiUV}{\ensuremath{\idx{\Pi}{UV}}}

\newcommand{\chh}[5]{\ensuremath{H_{#1}^{#2}(\textnormal{#3}|\textnormal{#4})_{#5}}}

\newcommand{\chmin}[3]{\chh{\textnormal{min}}{}{#1}{#2}{#3}}
\newcommand{\chmineps}[4]{\chh{\textnormal{min}}{#1}{#2}{#3}{#4}}

\newcommand{\chmax}[3]{\chh{\textnormal{max}}{}{#1}{#2}{#3}}
\newcommand{\chmaxeps}[4]{\chh{\textnormal{max}}{#1}{#2}{#3}{#4}}

\newcommand{\chvn}[3]{\chh{}{}{#1}{#2}{#3}}

\newcommand{\epsball}[2]{\ensuremath{\mathcal{B}^{#1}(#2)}}
\newcommand{\epsballpure}[2]{\ensuremath{\mathcal{B}_\textnormal{p}^{#1}(#2)}}

\theoremstyle{plain}
\newtheorem{lemma}{Lemma}
\newtheorem{theorem}[lemma]{Theorem}
\newtheorem{corollary}[lemma]{Corollary}
\theoremstyle{definition}
\newtheorem{definition}[lemma]{Definition}
\newtheorem{remark}[lemma]{Remark}

\newcommand{\dist}{\bar{D}}

\begin{document}

\title{Duality Between Smooth Min- and Max-Entropies}

\date{October 20, 2009}

\author{Marco \surname{Tomamichel}}
\email[]{marcoto@phys.ethz.ch}
\affiliation{Institute for Theoretical Physics, ETH Zurich, 8093
  Zurich, Switzerland.}
\author{Roger \surname{Colbeck}}
\email[]{colbeck@phys.ethz.ch}
\affiliation{Institute for Theoretical Physics, ETH Zurich, 8093
  Zurich, Switzerland.}
\affiliation{Institute of Theoretical Computer Science, ETH Zurich, 8092
  Zurich, Switzerland.}
\author{Renato \surname{Renner}}
\email[]{renner@phys.ethz.ch}
\affiliation{Institute for Theoretical Physics, ETH Zurich, 8093
  Zurich, Switzerland.}

\begin{abstract}

In classical and quantum information theory, operational quantities
  such as the amount of randomness that can be extracted from a given
  source or the amount of space needed to store given data are
  normally characterized by one of two entropy measures, called smooth
  min-entropy and smooth max-entropy, respectively. While both
  entropies are equal to the von Neumann entropy in certain special
  cases (e.g., asymptotically, for many independent repetitions of the
  given data), their values can differ arbitrarily in the general
  case.

  In this work, a recently discovered duality relation between
  (non-smooth) min- and max-entropies is extended to the smooth case. More precisely,
  it is shown that the smooth min-entropy of a system A conditioned on a
  system~B equals the negative of the smooth max-entropy of A conditioned
  on a purifying system~C. This result immediately implies that certain
  operational quantities (such as the amount of compression and the amount
  of randomness that can be extracted from given data) are related. Such
  relations may, for example, have applications in cryptographic security
  proofs.

\end{abstract}

\maketitle


\section{Introduction}

\label{sec:intro}

Entropies are used to quantitatively characterize
problems in quantum information processing and quantum cryptography.
In the case of many independent and identical instances of a task
(i.i.d.\ limit), the von Neumann entropy is the relevant measure. In
order to go beyond this restriction, the smooth min- and max-entropies
have been introduced.  The smooth min-entropy was introduced in order
to characterize randomness extraction.  It corresponds to the length
of uniform random string that can be generated from a partially
unifrom one~\cite{rennerkoenig05, renner05}.  The smooth max-entropy, on the
other hand, was introduced to characterize information reconciliation.
It gives the amount of communication required between two parties in
order that they can generate a perfectly correlated string from a
partially correlated one~\cite{rennerwolf05}.
Since their initial uses, these entropies have found applications in
many tasks (see for example~\cite{berta08, datta09}) and have been
shown to converge to the von Neumann entropy in the i.i.d.\
limit~\cite{renner05, tomamichel08}.

The smooth entropies can be defined as optimizations of the relevant
non-smooth quantities\,|\,the (non-smooth) min- and
max-entropies\,|\,over a set of nearby states. `Nearby' is specified
via a \emph{smoothing parameter}, the maximum distance from the
original state in an appropriate metric (for precise definitions, see
below).
Often, the smooth entropy is the correct measure when one accounts
for a small error tolerance, whereas the non-smooth entropy
characterizes the zero error case. In the case of privacy
amplification, for example, ideally one wants a protocol in which
two parties, Alice (A) and Bob (B), use a shared string about which an
eavesdropper (E) has partial information and form a secret key about
which E knows nothing. Unfortunately, such a stringent requirement
is usually unattainable. Instead, one tolerates a small probability
that the eavesdropper learns something about the key in order to
achieve a significant key length. In this case, the smooth
min-entropy of A given E characterizes the length of the key, with
the smoothing parameter dependent on the tolerable
error~\cite{rennerkoenig05}.

It has recently been discovered that the min- and max-entropies are
related~\cite{koenig08}. They are dual to one another in the sense
that for a pure state $\rhoABC$ on a tri-partite Hilbert space $\hA
\kron \hB \kron \hC$, the conditional min-entropy of A given B is the
negative of the conditional max-entropy of A given C, i.e.\
$\chmin{A}{B}{\rho} = -\chmax{A}{C}{\rho}$. In this work, we extend
the duality relation to the smooth min- and max-entropies. In order to
do this, a new method of smoothing is required: We propose measuring
the closeness of states used in the definition of the smooth entropies
in terms of a quantity which we call the \emph{purified
  distance}. This forms a metric on the set of sub-normalized states
(positive semi-definite operators with trace at most 1). When defined
in this way, the smooth min- and max-entropies satisfy various natural
properties such as invariance under local isometries and the data
processing inequality (that quantum operations cannot decrease
entropy). The duality not only simplifies many
derivations\footnote{Various relations for the min-entropy hold
  automatically for the max-entropy via the duality,
  e.g.\ the fact that conditioning on an additional system always
  reduces entropy (cf.\ Theorem~\ref{thm:data-proc}).}, but it
provides a connection between seemingly unrelated problems. In
particular, this means that randomness extraction and information
reconciliation can be characterized by the same entropy.

In the context of quantum key distribution, the new relation has the
following application. In order to generate a secure key, Alice and Bob need
to bound the smooth min-entropy of A conditioned on E. Our result
provides a way for them to generate this bound without access to the
eavesdropper's systems. In the worst case, the eavesdropper holds a purification of
the state of A and B. (The data processing inequality
(Theorem~\ref{thm:data-proc}) implies that if she does not, her
information about the key is strictly smaller.) Using the duality
relation, Alice and Bob obtain the desired bound on the smooth
min-entropy by estimating the smooth max-entropy of A given B.

There is an alternative method for going beyond i.i.d.\ in information
theory, known as the \emph{information spectrum method}
\cite{han02,hayashi03,nagaoka07}. Like for smooth entropies, there are
two principal quantities: the \emph{inf-spectral entropy rate} which
is related to the smooth min-entropy and the \emph{sup-spectral
  entropy rate} which is related to the smooth max-entropy
\cite{dattarenner08}. The results of this paper imply that a similar
duality relation holds for the spectral entropy rates.

The remainder of this paper is organized as follows. In
Section~\ref{sec:metrics} we introduce the purified distance and prove
that it is a metric on sub-normalized states. In
Section~\ref{sec:eps-ball} we use this metric to define a \emph{ball}
of states around a particular state. This ball is then used to define
the smooth conditional min- and max-entropies in
Section~\ref{sec:smooth} and to prove that they satisfy data
processing inequalities in Section~\ref{sec:dataproc}.

\section{Metrics on the Set of Sub-Normalized States}
\label{sec:metrics}

Let $\h$ be a finite-dimensional Hilbert space. We use $\linops{\h}$
and $\posops{\h}$ to denote the set of linear operators on $\h$ and
the set of positive semi-definite operators on $\h$, respectively.
We define the set of normalized quantum states by $\normstates{\h}
:= \{ \rho \in \posops{\h} : \tr\,\rho = 1 \}$ and the set of
sub-normalized states by $\subnormstates{\h} := \{ \rho \in
\posops{\h} : 0 < \tr\,\rho \leq 1 \}$. Note that $\linops{\h}
\supset \posops{\h} \supset \subnormstates{\h} \supset
\normstates{\h}$. Given a pure state $\ket{\phi} \in \h$, we use
$\phi = \proj{\phi}{\phi}$ to denote the corresponding projector in
$\posops{\h}$.

We start by introducing a generalization of the trace distance:
\begin{definition}
For $\rho, \tau \in \posops{\h}$, we define the \emph{generalized
  trace distance}
between $\rho$ and $\tau$ as
\begin{equation*}
\dist(\rho, \tau)\ :=\ \max \big\{ \tr\, \{ \rho - \tau \}_+,
\, \tr\, \{ \tau - \rho \}_+ \big\} \, ,
\end{equation*}
where $\{ X \}_+$ denotes the projection of $X$ onto its positive
eigenspace.
\end{definition}
In the case of normalized states, we have $\tr\, \{ \rho - \tau \}_+ =
\tr\, \{ \tau - \rho \}_+$ and recover the usual trace distance
$D(\rho, \tau) := \tr\, \{\rho - \tau\}_+$. The generalized trace
distance can alternatively be expressed in terms of the Schatten
$1$-norm $\norm{X}{1} = \tr\,\abs{X} = \tr\,\sqrt{X^\dagger X}$ as
\begin{equation*} \dist(\rho, \tau)\ =\ \frac{1}{2} \normbig{\rho -
    \tau}{1} + \frac{1}{2} \big| \tr\,\rho - \tr\,\tau \big| \,
\end{equation*} and it is easy to verify that it is a metric on
$\linops{\h}$. The trace distance has a physical interpretation as the
distinguishing advantage between two normalized states. In other
words, the probability $p_{\textrm{dist}}(\rho, \tau)$ of correctly
guessing which of two equiprobable states $\rho$ and $\tau$
is provided is upper bounded by \cite{nielsen00}
\begin{equation}
  \label{eqn:dist-adv}
  p_\textrm{dist}(\rho, \tau) \leq \frac{1}{2}
  \big(1 + D(\rho, \tau) \big).
\end{equation}

Various quantities derived from the fidelity $F(\rho, \tau) =
\norm{\sqrt{\rho} \sqrt{\tau}}{1}$ are used in the literature to
quantify the distance between normalized states.  Its generalization
to sub-normalized states satisfies $0 \leq F(\rho, \tau) \leq
\sqrt{\tr\,\rho}\, \sqrt{\tr\,\tau}$ and is monotonically increasing
under trace preserving completely positive maps (TP-CPMs), i.e.\
$F\big(\mathcal{E}(\rho), \mathcal{E}(\tau) \big) \geq F(\rho, \tau)$
for any TP-CPM $\mathcal{E}$ (cf.\ \cite{nielsen00}, Theorem 9.6).
Moreover, we will often use Uhlmann's theorem \cite{uhlmann85} which
states that, for any purification $\varphi$ of $\rho$, there exists a
purification $\vartheta$ of $\tau$ such that $F(\rho, \tau) =
F(\varphi, \vartheta) = \abs{\braket{\varphi}{\vartheta}}$.  The
fidelity is also symmetric in its arguments, i.e.\ $F(\rho, \tau) =
F(\tau, \rho)$.

For our argument, we need an alternative generalization of the
fidelity to sub-normalized states. The generalization is motivated by
the observation that sub-normalized states can be thought of as
normalized states on a larger space projected onto a subspace. We
write $\bar{\h} \supseteq \h$ if a Hilbert space $\h$ is embedded in
another Hilbert space $\bar{\h}$ and denote the projector onto $\h$ by
$\Pi$.
\begin{definition}
\label{def:generalized-fidelity}
  For $\rho, \tau \in \subnormstates{\h}$, we define the
  \emph{generalized fidelity} between $\rho$ and $\tau$ as
\begin{equation}
\label{eqn:generalized-fidelity}
\bar{F}(\rho, \tau)\ :=\ \sup_{\bar{\h} \supseteq \h}
\mathop{\sup_{\rhob,\, \taub\, \in\, \normstates{\bar{\h}}}}_{\rho =
\Pi\rhob\Pi,\, \tau =\Pi\taub\Pi}\!\!  F(\rhob, \taub) \, .
\end{equation}
\end{definition}

Note that $\bar{F}$ reduces to $F$ when at least one state is
normalized. This can be seen from the following alternative expression
for $\bar{F}$:
\begin{lemma}
  \label{lemma:alt-generalized-fidelity} Let $\rho, \tau \in
  \subnormstates{\h}$. Then,
  \begin{equation*}
    \label{eqn:alt-generalized-fidelity} \bar{F}(\rho, \tau) = F(\rhoh,
    \tauh) = F(\rho, \tau) + \sqrt{(1 - \tr\,\rho)(1 - \tr\,\tau)}\, ,
  \end{equation*}
  where $\rhoh := \rho \oplus (1\!-\!\tr\,\rho)$ and $\tauh :=
  \tau \oplus (1\!-\!\tr\,\tau)$.
\end{lemma}
\begin{proof}
  Let $\bar{\h}, \rhob$ and $\taub$ be any combination of Hilbert
  space and states that are candidates for the maximization
  in~\eqref{eqn:generalized-fidelity}.  Let $\mathcal{E} : \bar{\h} \to
  \bar{\h}$ be the pinching $\mathcal{E} : \rhob \mapsto \Pi \rhob \Pi +
  \Pi_\perp \rhob \Pi_\perp$, where $\Pi$ is the projector onto $\h$ and
  $\Pi_\perp=\idi{$\bar{\h}$}-\Pi$ its orthogonal complement. Hence,
  \begin{align*}
    F(\rhob, \taub) &\leq F \big( \mathcal{E}(\rhob), \mathcal{E}(\taub) \big) \\
    &= F(\rho, \tau) + F\big( \Pi_\perp\rhob \Pi_\perp, \Pi_\perp\taub \Pi_\perp \big) \\
    &\leq F(\rho, \tau)  + \sqrt{(1-\tr\,\rho)(1-\tr\,\tau)} \, .
  \end{align*}
  It is easy to verify that the upper bound is achieved by $\hat{\h} = \h
  \oplus \mathbb{C}$, $\rhoh$ and $\tauh$.
\end{proof}
%

We define a metric based on the fidelity, analogously to the one
proposed in~\cite{nielsen04, rastegin06}\footnote{The quantity
  $C(\rho, \tau) = \sqrt{1 - F^2(\rho, \tau)}$ is introduced
  in~\cite{nielsen04}, where they also show that it is a metric on
  $\normstates{\h}$. In~\cite{rastegin06} the quantity is called sine
  distance and some of its properties are explored.}:
\begin{definition}
  For $\rho, \tau \in \subnormstates{\h}$, we define the
  \emph{purified distance} between $\rho$ and $\tau$ as
\begin{equation*}
P(\rho, \tau) := \sqrt{1 - \bar{F}(\rho, \tau)^2}
\end{equation*}
\end{definition}
The name is motivated by the fact that, for normalized states $\rho,
\tau \in \normstates{\h}$, we can write $P(\rho, \tau)$ as the minimum
trace distance between purifications $\ket{\varphi}$ of $\rho$ and
$\ket{\vartheta}$ of $\tau$.  More precisely, using Uhlmann's
theorem~\cite{uhlmann85}, we have
\begin{align*}
P(\rho, \tau)
&= \sqrt{1 - {F(\rho, \tau)}^2}
= \sqrt{1 - \max_{\varphi, \vartheta} \Abs{\braket{\varphi}{\vartheta}}^2} \\
&= \min_{\varphi, \vartheta} \sqrt{1 - \Abs{\braket{\varphi}{\vartheta}}^2}
= \min_{\varphi, \vartheta} \dist(\varphi, \vartheta) \, .
\end{align*}

\begin{lemma}
\label{lemma:metric}
The purified distance $P$ is a metric on $\subnormstates{\h}$.
\end{lemma}

\begin{proof}
  Let $\rho, \tau$ and $\sigma$ be any states in $\subnormstates{\h}$.
  The condition $P(\rho, \tau) = 0 \iff \rho = \tau$ can be verified
  by inspection, and symmetry $P(\rho, \tau) = P(\tau, \rho)$ follows
  from the symmetry of the fidelity.

  It remains to show the triangle inequality $P(\rho, \tau) \leq
  P(\rho, \sigma) + P(\sigma, \tau)$. Using
  Lemma~\ref{lemma:alt-generalized-fidelity}, the
  generalized fidelities between $\rho$, $\tau$ and $\sigma$ can be
  expressed as fidelities between the corresponding extensions
  $\rhoh$, $\tauh$ and $\sigmah$. Furthermore, we use
  Uhlmann's theorem to introduce purifications $\ket{r}$ of $\rhoh$,
  $\ket{s}$ of $\sigmah$ and $\ket{t}$ of $\tauh$ such that $F(\rhoh,
  \sigmah) = \abs{\braket{r}{s}}$, $F(\sigmah, \tauh) =
  \abs{\braket{s}{t}}$ and $F(\rhoh, \tauh) \geq \abs{\braket{r}{t}}$.
  Hence,
  \begin{align*}
    P(\rho, \sigma) + P(\sigma, \tau) &= P(r, s) + P(s, t) \\
    & = D(r, s) + D(s, t) \\
    &\geq D(r, t) = P(r, t) \geq P(\rho, \tau) \, ,
  \end{align*}
  where we have used the triangle inequality for the trace
  distance.
\end{proof}

The following lemma gives lower and upper bounds to the purified distance
in terms of the generalized trace distance.
\begin{lemma}
  \label{lemma:metric-bounds}
  Let $\rho, \tau \in \subnormstates{\h}$. Then $$\dist(\rho, \tau) \leq
  P(\rho, \tau) \leq \sqrt{2 \dist(\rho, \tau)} \, .$$
\end{lemma}
\begin{proof}
  We express the quantities using the normalized extensions $\rhoh$
  and $\tauh$ of Lemma~\ref{lemma:alt-generalized-fidelity} to get
  \begin{align*}
    &P(\rho, \tau) = \sqrt{1 - F(\rhoh, \tauh)^2} \geq D(\rhoh, \tauh) = \dist(\rho, \tau)\, \quad  \textrm{and} \\
    &P(\rho, \tau)^2 = 1 - F(\rhoh, \tauh)^2\leq 1 - \big(1 - D(\rhoh, \tauh) \big)^2 \leq 2 \dist(\rho, \tau) \, ,
  \end{align*}
  where we have made use of $1 - F(\rhoh, \tauh) \leq D(\rhoh, \tauh)
  \leq \sqrt{1 - F(\rhoh, \tauh)^2}$~(see e.g.\ \cite{nielsen00},
  Section 9.2.3).
\end{proof}

A useful property of the purified distance is that it does not increase
under simultaneous application of a quantum operation on both states.
We consider the class of trace non-increasing CPMs, which includes projections.
\begin{lemma}
\label{lemma:purified-monotonicity}
Let $\rho, \tau \in \subnormstates{\h}$ and $\mathcal{E}$ be a trace
non-increasing CPM. Then,
$P(\rho, \tau) \geq P\big(
\mathcal{E}(\rho), \mathcal{E}(\tau) \big)$.
\end{lemma}
\begin{proof}
  Note that a trace non-increasing CPM $\mathcal{E} : \posops{\h}
  \to \posops{\h'}$ can be decomposed into an isometry $U: \h \to \h'
  \kron \h''$ followed by a projection $\Pi \in \posops{\h' \kron
    \h''}$ and a partial trace over $\h''$ (see, e.g.\
  \cite{nielsen00}, Section~8.2).  The isometry and the partial trace
  are TP-CPMs and, hence, it suffices to show that $\bar{F}(\rho,
  \tau) \leq \bar{F}\big(\mathcal{E}(\rho), \mathcal{E}(\tau) \big)$ for TP-CPMs
  and projections.

  First, let $\mathcal{E}$ be trace preserving.  Using
  Lemma~\ref{lemma:alt-generalized-fidelity} and the
  monotonicity under TP-CPMs of the fidelity, we see that
  \begin{align*}
    \bar{F}(\rho, \tau) &= F(\rho, \tau) +
    \sqrt{(1-\tr\,\rho)(1-\tr\,\tau)} \\
    &\leq F\big(\mathcal{E}(\rho), \mathcal{E}(\tau)\big) +
    \sqrt{(1-\tr\,\rho)(1-\tr\,\tau)} \\
    &= \bar{F}\big(\mathcal{E}(\rho),
    \mathcal{E}(\tau) \big)\, .
  \end{align*}

  Next, consider a projection $\Pi \in \posops{\h}$ and the CPM
  $\mathcal{E}: \rho \mapsto \Pi\rho\Pi$. Following
  Definition~\ref{def:generalized-fidelity}, we write $\bar{F}(\rho,
  \tau) = \sup\, F(\rhob, \taub)$, where the supremum is taken over
  all extensions $\{ \bar{\h}, \rhob, \taub \}$ of $\{ \h, \rho, \tau
  \}$. Since all extensions of $\{ \h, \rho, \tau \}$ are also
  extensions of $\big\{ \supp{\Pi}, \Pi \rho\Pi, \Pi\tau\Pi \big\}$,
  we find $\bar{F}\big(\Pi\rho\Pi, \Pi\tau\Pi \big) \geq \bar{F}(\rho,
  \tau)$.
\end{proof}

The main advantage of the purified distance over the trace distance is that we can always find extensions and purifications without increasing the distance.
\begin{lemma}
\label{lemma:purified-Uhlmann}
Let $\rho, \tau \in \subnormstates{\h}$, $\h' \iso \h$ and $\varphi
\in \h \kron \h'$ be a purification of $\rho$. Then, there exists a purification $\vartheta \in \h \kron \h'$ of $\tau$ with $P(\rho, \tau) = P(\varphi, \vartheta)$.
\end{lemma}
\begin{proof}
We use Uhlmann's theorem to choose $\vartheta \in \h \kron \h'$ such that $F(\rho, \tau) = F(\varphi, \vartheta)$ and, thus, $P(\rho, \tau) = P(\varphi, \vartheta)$.
\end{proof}

\begin{corollary}
\label{cor:purified-ext}
Let $\rho, \tau \in \subnormstates{\h}$ and $\rhob \in
\subnormstates{\h \kron \h'}$ be an extension of $\rho$. Then, there
exists an extension $\taub \in \subnormstates{\h \kron \h'}$ of $\tau$
with $P(\rho, \tau) = P(\rhob, \taub)$.
\end{corollary}
\begin{proof}
  Let $\h'' \iso \h \kron \h'$ be an auxiliary Hilbert space and
  $\varphi \in \h \kron \h' \kron \h''$ be a purification of
  $\rhob$. We introduce a purification $\vartheta \in \h \kron \h'
  \kron \h''$ of $\tau$ with $P(\varphi, \vartheta) = P(\rho, \tau)$
  using Lemma~\ref{lemma:purified-Uhlmann} and $\taub =
  \ptr{\h''}\,\vartheta$. However, due to
  Lemma~\ref{lemma:purified-monotonicity}, we have $P(\varphi,
  \vartheta) \geq P(\rhob, \taub) \geq P(\rho, \tau)$, which implies
  that all three distances must be equal.
\end{proof}

\section{The $\eps$-neighborhood induced by $P$}
\label{sec:eps-ball}

The $\eps$-smooth min-entropy of a state $\rho$ is usually defined
as a maximization of the min-entropy over a set of states that are
$\eps$-close to $\rho$. Various definitions of such
sets\,|\,subsequently called $\eps$-balls\,|\,have appeared in the
literature. None of the existhing definitions exhibit the
following two properties that are of particular importance in the
context of smooth conditional min- and max-entropies: Firstly, the
smooth entropies should be independent of the Hilbert spaces used to
represent the state. In particular, embedding the density operator
into a larger Hilbert space should leave the smooth entropies
unchanged. This can be achieved by allowing sub-normalized states in
the $\eps$-balls. Secondly, it will be important that we can define
a ball of pure states that contains purifications of all the states
in the $\eps$-ball. This allows us to establish the duality relation
between smooth min- and max-entropies and is achieved by using a
fidelity-based metric to determine $\eps$-closeness. The following
ball possesses both of the above properties:
\begin{definition}
  \label{def:eps-ball}
  Let $\eps \geq 0$ and $\rho \in \subnormstates{\h}$ with
  $\sqrt{\tr\,\rho} > \eps$. Then, we define
  an \emph{$\eps$-ball} in $\h$ around $\rho$ as
  \begin{equation*}
    \label{eqn:def-eps-ball}
    \epsball{\eps}{\h ; \rho} := \{ \tau \in \subnormstates{\h} :
    P(\tau, \rho) \leq \eps \} \, .
  \end{equation*}
  We also define $\epsballpure{\eps}{\h; \rho} := \{ \tau \in
  \epsball{\eps}{\h; \rho} : \rank\, \tau = 1 \}$, i.e.\ an $\eps$-ball
  of pure states around $\rho$.
\end{definition}

We now prove some properties of the $\eps$-ball that will be important
for our later discussion of smooth conditional min- and
max-entropies. Properties~i)--iv) clarify what we mean by an
$\eps$-ball around $\rho$. Property~v) ensures that states in the ball
remain in the ball after applying an isometry, while Properties~vi)
and vii) relate to how the $\eps$-balls change under partial trace and
purification.  These will be particularly relevant for the duality
relation between the smooth min- and max-entropies.

\begin{enumerate}


\item[i)] The set $\epsball{\eps}{\h; \rho}$ is compact and convex.

\begin{proof}
  The set is closed and bounded, hence compact.
  For convexity, we require that, for any $\lambda \in [0, 1]$
  and $\sigma, \tau \in \epsball{\eps}{\h; \rho}$, the state
  $\omega := \lambda \sigma + (1 - \lambda) \tau$ is also in
  $\epsball{\eps}{\h; \rho}$.
  We define $\omegah = \omega \oplus (1\!-\!\tr\,\omega)$ and
  analogously $\rhoh$, $\sigmah$ and $\tauh$. By assumption we have
  $F(\sigmah, \rhoh) \geq \sqrt{1 - \eps^2}$ and $F(\tauh, \rhoh) \geq
  \sqrt{1 - \eps^2}$.  We use the concavity of the fidelity~(cf.\
  \cite{nielsen00}, Section 9.2.2) to find
\begin{align*}
  P(\omega, \rho) &= \sqrt{1 - F(\omegah, \rhoh)^2} \\
  &= \sqrt{1 - F(\lambda \sigmah + (1\!-\!\lambda) \tauh, \rhoh)^2}  \\
  &\leq \sqrt{1 - \big( \lambda F(\sigmah, \rhoh) + (1\!-\!\lambda) F(\tauh, \rhoh) \big)^2}
\leq \eps \, .
\end{align*}
Therefore, $\omega \in \epsball{\eps}{\h; \rho}$, as required.
\end{proof}


\item[ii)] Normalized states in $\epsball{\eps}{\h; \rho}$ are not
  distinguishable from $\rho$ with probability more than
  $\frac{1}{2}(1 + \eps)$.
\begin{proof}
  By Lemma~\ref{lemma:metric-bounds}, $\tau \in \epsball{\eps}{\h;
    \rho}$ implies $\dist(\tau, \rho) \leq P(\tau, \rho) \leq \eps$. The
  statement then follows from~\eqref{eqn:dist-adv}.
\end{proof}


\item[iii)] The ball grows monotonically in the smoothing parameter $\eps$.
  Furthermore, $\epsball{0}{\h; \rho} = \{ \rho \}$.

\item[iv)] The $\eps$-balls are symmetric and satisfy a triangle
  inequality.  In other words, we have
\begin{align*}
&\tau \in \epsball{\eps}{\h; \rho}\! \iff\! \rho \in
\epsball{\eps}{\h;
  \tau}   \quad \textrm{and} \\
&\tau \in \epsball{\eps}{\h; \rho} \wedge \sigma \in
\epsball{\eps'}{\h; \tau} \!\implies\! \sigma \in \epsball{\eps +
\eps'}{\h; \rho} \, .
\end{align*}

\begin{proof}
These properties follow directly from the fact that $P$ is a metric (cf.~Lemma~\ref{lemma:metric}).
\end{proof}


\item[v)] The $\eps$-balls are invariant under isometries. Let $U: \h
  \to \h'$ be an isometry, then
\begin{equation*}
\label{eqn:ball/iso-forward}
\tau \in \epsball{\eps}{\h; \rho} \implies U \tau\, U^\dagger \in \epsball{\eps}{\h'; U \rho\, U^\dagger}
 \, .
\end{equation*}
Conversely,
if $\Pi$ is the projector onto the image of $U$, then
\begin{equation*}
\label{eqn:ball/iso-backward}
\sigma \in \epsball{\eps}{\h'; U \rho\, U^\dagger} \implies U^\dagger \Pi \sigma \Pi U \in \epsball{\eps}{\h; \rho}\, .
\end{equation*}

\begin{proof}
  This property follows from Lemma~\ref{lemma:purified-monotonicity}
  and the fact that $\rho \mapsto U \rho U^\dagger$ and $\rho \mapsto
  U^\dagger \Pi \rho \Pi U$ are trace non-increasing
  CPMs.\footnote{Note that the state $U^\dagger \Pi \sigma \Pi U$ is
    not necessarily normalized. Definitions of $\eps$-balls that do
    not allow sub-normalized states will not be invariant under
    isometries in the sense proposed here. This property will be used
    to show that the smooth min- and max-entropies are invariant under
    local isometries in Lemmas~\ref{lemma:minepsiso}
    and~\ref{lemma:maxepsiso}.}
\end{proof}


\item[vi)] The $\eps$-balls are monotone under partial trace.  More
  precisely, let $\h'$ be a Hilbert space and $\ptr{\h'}$
  the partial trace over $\h'$, then
\begin{equation*}
\tau \in \epsball{\eps}{\h \kron \h'; \rho} \implies \ptr{\h'}{\tau} \in \epsball{\eps}{\h;\, \ptr{\h'}\,\rho}\, .
\end{equation*}

\begin{proof}
This property is a direct consequence of
Lemma~\ref{lemma:purified-monotonicity} and the fact that the partial trace is a TP-CPM.
\end{proof}


\item[vii)] On a sufficiently large Hilbert space, there exists a purification of the $\eps$-ball in the following sense: Let $\h'$ be a Hilbert space with $\dim \h' \geq \dim \h$ and $\varphi \in \h \kron \h'$, then
\begin{align*}
  &\tau \in \epsball{\eps}{\h;\, \ptr{$\h'$}\,\varphi} \implies \nonumber\\
  & \qquad\exists\, \vartheta \in \epsballpure{\eps}{\h \kron \h';
    \varphi}\ \textrm{s.t.}\ \tau = \ptr{\h'}\,\vartheta\, .
\end{align*}

\begin{proof}
  This property follows from
  Lemma~\ref{lemma:purified-Uhlmann}.\footnote{ This property of the
    $\eps$-ball is due to our use of a fidelity-based metric. In
    particular, it does not hold for an $\eps$-ball based on the trace
    distance, such as $\{ \tau \in \subnormstates{\h} : \dist(\rho,
    \tau) \leq \eps \}$.  We will use this in
    Lemma~\ref{lemma:epsballineq} to show that the duality
    relation for the smooth entropy holds.}
\end{proof}

\end{enumerate}

\section{Smooth Conditional Min- and Max-Entropies}
\label{sec:smooth}

In this section we define smooth min- and max-entropies and discuss
some of their properties that follow from our definition of the
$\eps$-ball. In particular, the smooth entropies defined in the
following can be seen as optimizations of the corresponding non-smooth
entropies over an $\eps$-ball of states (Definition~\ref{def:mineps}
and Lemma~\ref{lemma:maxepsball}). Moreover, they are invariant under
local isometries (Lemmas~\ref{lemma:minepsiso}
and~\ref{lemma:maxepsiso}) and satisfy a duality relation
(Definition~\ref{def:maxeps}).\footnote{For convenience of exposition,
  we will define the smooth max-entropy as the dual of the smooth
  min-entropy (Definition~\ref{def:maxeps}) and then prove that this
  definition is equivalent to an optimization over an $\eps$-ball of
  states of the non-smooth max-entropy (Lemma~\ref{lemma:maxepsball}).}

In the following, we assume that $\eps$ is much smaller than the trace
of all involved states as is predominantly the case in applications. Indices
are used to denote multi-partite Hilbert spaces, e.g.\ $\hAB = \hA
\kron \hB$ and to denote the different marginal states of
multi-partite systems. We often do not mention explicitly when a
partial trace needs to be taken, e.g.\ if $\rhoAB \in
\subnormstates{\hAB}$ is given, then $\rhoA := \ptr{B}\,\rhoAB$ is
also implicitly defined.

We define the min-entropy:
\begin{definition}
  Let $\rhoAB \in \subnormstates{\hAB}$, then the \emph{min-entropy}
  of A conditioned on B of $\rhoAB$ is defined as
  \begin{equation*} \label{eqn:min-entropy-cond}
    \chmin{A}{B}{\rho} := \!\! \max_{\sigmaB \in \normstates{\hB}}
    \! \sup\, \{ \lambda \in \mathbb{R} : 2^{-\lambda}\, \idA
    \kron \sigmaB \geq \rhoAB \} \, .
  \end{equation*}
\end{definition}
We now use the $\eps$-ball $\epsball{\eps}{\h; \rho}$ to define a
smoothed version of the min-entropy\footnote{Note that we drop $\h$
  when the Hilbert space is clear from the indices of the state.}:
\begin{definition}
  \label{def:mineps}
  Let $\eps \geq 0$ and $\rhoAB \in \subnormstates{\hAB}$,
  then the \emph{$\eps$-smooth min-entropy} of A conditioned on B of
  $\rhoAB$ is defined as
  \begin{equation*}
    \chmineps{\eps}{A}{B}{\rho} := \!\! \max_{\rhotAB \in
      \epsball{\eps}{\rhoAB}} \! \chmin{A}{B}{\rhot}\, .
  \end{equation*}
\end{definition}
The quantity is monotonically increasing in $\eps$ due to
Property~iii) in Section~\ref{sec:eps-ball} and we recover the
non-smooth entropy by $\chmineps{0}{A}{B}{\rho} = \chmin{A}{B}{\rho}$.
Continuity of the smooth min-entropy as a function of the state is
shown in Appendix~\ref{app:tech}.

The smooth min-entropy is independent of the Hilbert spaces used to
represent the density operator locally, as the following lemma
shows:
\begin{lemma}
  \label{lemma:minepsiso}
  Let $\eps \geq 0$, $\rhoAB \in \subnormstates{\hAB}$ and $U: \hA \to
  \hC$ and $V: \hB \to \hD$ be two isometries with $\tauCD := (U \kron
  V) \rhoAB (U^\dagger \kron V^\dagger)$, then
  \begin{equation*}
    \chmineps{\eps}{A}{B}{\rho} = \chmineps{\eps}{C}{D}{\tau} \,.
  \end{equation*}
\end{lemma}
\begin{proof}
  First note that the $\eps$-smooth min-entropy $\chmineps{\eps}{A}{B}{}$ can be written as
  \begin{equation}
  \label{eqn:mineps-alt}
  \chmineps{\eps}{A}{B}{\rho}\ = \!\! \max_{\rhotAB\,\in\,\epsball{\eps}{\rhoAB}} \!
  \mathop{\max_{\sigmaB\,\in\,\posops{\hB}}}_{\rhotAB\,\leq\,\idA
    \kron \sigmaB} -\log \tr\,\sigmaB \, ,
  \end{equation}
  where $\log$ denotes the binary logarithm. Now, we let $\rhobAB \in
  \epsball{\eps}{\rhoAB}$ and $\sigmabB \in \posops{\hB}$ be the pair
  of states that maximize this expression,
  i.e.~$\chmineps{\eps}{A}{B}{\rho} = - \log \tr\,\sigmabB$. Then
  $\rhobAB \leq \idA \kron \sigmabB$ implies
  $$ \underbrace{\, (U \kron V) \rhobAB (U^\dagger \kron V^\dagger) \,
    }_{\, =: \,\taubCD} \leq U U^\dagger \kron V \sigmabB V^\dagger \leq
    \idC \kron \underbrace{\, V \sigmabB V^\dagger\, }_{\, =:\,\omegabD} \, .$$
  The $\eps$-ball is invariant under isometries (cf.~Property~v)) and,
  therefore, the pair $\taubCD \in \epsball{\eps}{\tauCD}$ and
  $\omegabD \in \posops{\hD}$ is a candidate for the optimization in
  $\chmineps{\eps}{C}{D}{\tau}$. We bound
  \begin{align*}
    \chmineps{\eps}{C}{D}{\tau} \geq - \log \tr\,\omegabD
    = - \log \tr\,\sigmabB
    = \chmineps{\eps}{A}{B}{\rho} \, .
  \end{align*}

  The argument in the reverse direction is similar. Let $\tautCD
  \in\epsball{\eps}{\tauCD}$ and $\omegatD \in \posops{\hD}$ be the
  pair that maximizes $\chmineps{\eps}{C}{D}{\tau}$. Moreover, we
  introduce $\PiUV = \PiU \kron\PiV$, where $\PiU$ and $\PiV$ are the
  projectors onto the image of $U$ and $V$ respectively. Then $\tautCD
  \leq \idC \kron \omegatD$ implies
  $$ \underbrace{\, (U^\dagger \kron V^\dagger) \PiUV \tautCD \PiUV
    (U \kron V)\, }_{\, =:\ \rhotAB \, } \leq \idA \kron
    \underbrace{\, V^\dagger \PiV \omegatD \PiV V \, }_{\, =:\
      \sigmatB }\, .$$
  The pair $\rhotAB \in \epsball{\eps}{\rhoAB}$ and $\sigmatB \in
  \mathcal{P}(\hB)$ is a candidate for the optimization in
  $\chmineps{\eps}{A}{B}{\rho}$ (cf.~Property~v)) and we get
  \begin{align*}
    \chmineps{\eps}{A}{B}{\rho} &\geq - \log
    \tr\,\sigmatB = - \log \Trace{\idx{\Pi}{V} \omegatD} \\
    &\geq - \log \tr\,\omegatD = \chmineps{\eps}{C}{D}{\tau} \, .
  \end{align*}
  We thus conclude that $\chmineps{\eps}{A}{B}{\rho} =
  \chmineps{\eps}{C}{D}{\tau}$.
\end{proof}

We next define the dual of the smooth min-entropy, the smooth
max-entropy:
\begin{definition}
  \label{def:maxeps}
  Let $\eps \geq 0$,
  $\rhoAB \in \subnormstates{\hAB}$ and $\rhoABC \in
  \subnormstates{\hABC}$ an arbitrary purification of $\rhoAB$, then
  the \emph{$\eps$-smooth max-entropy} of A conditioned on B of $\rhoAB$ is defined as
  \begin{equation}
    \label{eqn:duality}
    \chmaxeps{\eps}{A}{B}{\rho} := -\chmineps{\eps}{A}{C}{\rho}\, .
  \end{equation}
\end{definition}
The quantity is well-defined since all purifications of $\rhoAB$ are
equivalent up to an isometry on the purifying space $\hC$, which
does not change $\chmineps{\eps}{A}{C}{}$ as
Lemma~\ref{lemma:minepsiso} shows. The non-smooth max-entropy is
given by $\chmax{A}{B}{} := \chmaxeps{0}{A}{B}{}$. An alternative
expression for the max-entropy was given in~\cite{koenig08}:
\begin{equation}
  \label{eqn:max-alt}
  \chmax{A}{B}{\rho} := \max_{\sigmaB \in \normstates{\hB}} \log F
  \big( \rhoAB, \idA \kron \sigmaB \big)^2 \, .
\end{equation}

The smooth max-entropy is independent of the Hilbert spaces used to
represent the density operator locally:
\begin{lemma}
  \label{lemma:maxepsiso}
  Let $\eps \geq 0$, $\rhoAB \in \subnormstates{\hAB}$, and $U: \hA
  \to \hC$ and $V: \hB \to \hi{D}$ be two
  isometries with $\tauCD := (U \kron V) \rhoAB (U^\dagger \kron
  V^\dagger)$, then
  \begin{equation*}
    \chmaxeps{\eps}{A}{B}{\rho} = \chmaxeps{\eps}{C}{D}{\tau} \,.
  \end{equation*}
\end{lemma}
\begin{proof}
  Let $\rhoABE$ be a purification of $\idx{\rho}{AB}$,
  then $\tauCDE = (U \kron V \kron \idE) \rhoABE (U^\dagger \kron
  V^\dagger \kron \idE)$ is a purification of $\tauCD$. Thus,
  \begin{align*}
    \chmaxeps{\eps}{A}{B}{\rho} &= - \chmineps{\eps}{A}{E}{\rho} \\
    &= -\chmineps{\eps}{C}{E}{\tau} = \chmaxeps{\eps}{C}{D}{\tau} \, .
  \end{align*}
\end{proof}
The $\eps$-smooth max-entropy can also be written as an optimization
over an $\eps$-ball of states:
\begin{lemma}
  \label{lemma:maxepsball}
  Let $\eps \geq 0$ and $\rhoAB \in \subnormstates{\hAB}$, then
  \begin{equation*}
    \chmaxeps{\eps}{A}{B}{\rho} = \!\! \min_{\rhotAB \in
    \epsball{\eps}{\rhoAB}} \! \chmax{A}{B}{\rhot}\, .
  \end{equation*}
\end{lemma}
In order to prove the above lemma, we characterize the $\eps$-ball
in terms of an $\eps$-ball on the purified space. The following
lemma follows directly from Properties~vi) and~vii) in
Section~\ref{sec:eps-ball} and will be used repeatedly:
\begin{lemma}
  \label{lemma:epsballineq} Let $\rho \in
  \subnormstates{\h}$ and $\phi \in \subnormstates{\h \kron \h'}$ be a
  purification of $\rho$, then
  \begin{equation*}
    \label{eqn:epsballineq} \epsball{\eps}{\h;\rho} \supseteq
    \{ \rhot \in \subnormstates{\h} : \exists\, \phit \in
    \epsballpure{\eps}{\h \kron \h';\phi} \ \textrm{s.t.}\ \rhot =
    \ptr{\h'}{\phit} \}
  \end{equation*}
  and the two sets are identical if the Hilbert space dimensions
  satisfy $\dim \h' \geq \dim \h$.
\end{lemma}
\begin{proof}[Proof of Lemma~\ref{lemma:maxepsball}]
  Let $\rhoABC \in \subnormstates{\hABC}$ be a purification of
  $\rhoAB$ with $\dim \hC \geq \dim \hAB$. Then, using
  Lemma~\ref{lemma:epsballineq} as well as Definition~\ref{def:maxeps}
  twice each, we have
  \begin{align*}
    \chmaxeps{\eps}{A}{B}{\rho} &= - \max_{\rhotAC \in
      \epsball{\eps}{\rhoAC}} \chmin{A}{C}{\rhot}\\
    &\leq \min_{\rhotABC \in \epsballpure{\eps}{\rhoABC}}
    -\chmin{A}{C}{\rhot}\\
    &= \min_{\rhotAB \in \epsball{\eps}{\rhoAB}}
    \chmax{A}{B}{\rhot} \, .
  \end{align*}

  To show the other direction, we choose a purification $\rhoABC \in
  \subnormstates{\hABC}$ and embed $\hB$ into a larger space
  $\hi{B$'$}$ such that $\dim \hi{B$'$} \geq \dim \hi{AC}$. We define
  $\idx{\rho}{AB$'$C}$ as the embedding of $\rhoABC$ into
  $\hi{AB$'$C}$. Furthermore, for each $\idx{\tilde{\rho}}{AB$'$C} \in
  \epsballpure{\eps}{\idx{\rho}{AB$'$C}}$ we construct
  $\idx{\rhob}{AB$'$C} = (\idA \kron \PiB \kron \idC)
  \idx{\rhot}{AB$'$C} (\idA \kron \PiB \kron \idC)$, where $\PiB$ is
  the projector onto $\hB$, such that $\idx{\rhob}{AB$'$C}$ remains in
  $\epsballpure{\eps}{\idx{\rho}{AB$'$C}}$ and has support on $\hABC$.
  The set $\{ \PiB, \idi{B$'$} - \PiB \}$ describes a measurement on
  $\hi{B$'$}$ and the post-measurement state $\rhobAC \in
  \subnormstates{\hAC}$ satisfies $\rhobAC \leq \rhotAC$. Hence,
  $\chmin{A}{C}{\rhot} \leq \chmin{A}{C}{\bar{\rho}}$.

  We may now write (again using Lemma~\ref{lemma:epsballineq} as well
  as Definition~\ref{def:maxeps} twice each):
  \begin{align*}
    \chmaxeps{\eps}{A}{B}{\rho}
    &= - \max_{\rhotAC \in \epsball{\eps}{\rhoAC}} \chmin{A}{C}{\rhot}\\
    &= \min_{\idx{\rhot}{AB$'$C} \in \epsballpure{\eps}{\idx{\rho}{AB$'$C}}} -
    \chmin{A}{C}{\rhot}\\
    &\geq \min_{\rhobABC \in \epsballpure{\eps}{\rhoABC}} - \chmin{A}{C}{\bar{\rho}} \\
    &\geq \min_{\rhobAB \in \epsball{\eps}{\rhoAB}} \chmax{A}{B}{\bar{\rho}} \, ,
  \end{align*}
  which concludes the proof.
\end{proof}

\section{Data-Processing Inequalities}
\label{sec:dataproc}

As an example of an application of the duality between smooth
conditional min- and max-entropies, we consider data-processing
inequalities for the two entropies. 

We expect measures of uncertainty about the system A given side
information B to be non-decreasing under local physical operations
applied to the B system. Here, we show that this is indeed the case
for $\chmineps{\eps}{A}{B}{}$ and $\chmaxeps{\eps}{A}{B}{}$. The
most general physical operations are modeled by TP-CPMs and we
denote by $\opidA$ the identity TP-CPM on $\posops{\hA}$.
\begin{theorem}
  \label{thm:data-proc}
  Let $\eps \geq 0$, $\rhoAB \in \subnormstates{\hAB}$ and
  $\mathcal{E}: \posops{\hB} \to \posops{\hD}$ be a
  TP-CPM with $\tauAD := \big( \opidA \kron \mathcal{E} \big)
  (\rhoAB)$, then
  \begin{align*}
    &\chmineps{\eps}{A}{B}{\rho} \leq \chmineps{\eps}{A}{D}{\tau} \quad
    \textrm{and} \\
    &\chmaxeps{\eps}{A}{B}{\rho} \leq \chmaxeps{\eps}{A}{D}{\tau}\, .
  \end{align*}
\end{theorem}
\begin{proof}
  For a sufficiently large Hilbert space $\hR$, the TP-CPM
  $\mathcal{E}$ can be decomposed into an isometry $U: \hB \to
  \hi{DR}$ followed by a partial trace over $\hi{R}$ (see
  e.g.~\cite{nielsen00}). The invariance of the two quantities under
  local isometries was established in Lemmas~\ref{lemma:minepsiso}
  and~\ref{lemma:maxepsiso}, so it remains to show that the quantities
  are non-decreasing under partial trace\footnote{This property is
    sometimes referred to as strong sub-additivity of the smooth min-
    and max-entropies. This is due to the fact that $\chvn{A}{BC}{}
    \leq \chvn{A}{B}{}$ is equivalent to the strong sub-additivity of the von
    Neumann entropy.}, i.e.\ the inequalities
  \begin{align*}
    &\chmineps{\eps}{A}{DR}{\tau} \leq \chmineps{\eps}{A}{D}{\tau}
    \quad \textrm{and} \\
    &\chmaxeps{\eps}{A}{DR}{\tau} \leq
    \chmaxeps{\eps}{A}{D}{\tau} \, .
  \end{align*}
  We first consider the inequality for the smooth min-entropy.
  Let $\idx{\taut}{ADR} \in \epsball{\eps}{\idx{\tau}{ADR}}$ and
  $\idx{\sigma}{DR} \in \posops{\hi{DR}}$ be the pair that optimizes
  the expression in~\eqref{eqn:mineps-alt} for
  $\chmineps{\eps}{A}{DR}{\tau}$, then
  $$\idx{\taut}{ADR} \leq \idA \kron \idx{\sigma}{DR} \implies
  \idx{\taut}{AD} \leq \idA \kron \idx{\sigma}{D}\,$$
  and, due to Property~vi), the pair $\idx{\taut}{AD} \in
  \epsball{\eps}{\idx{\tau}{AD}}$ and $\idx{\sigma}{D} \in
  \posops{\hi{D}}$ is a candidate for the optimization in
  $\chmineps{\eps}{A}{D}{\tau}$. Thus, $\chmineps{\eps}{A}{DR}{\tau}
  \leq \chmineps{\eps}{A}{D}{\tau}$.

  For the smooth max-entropy, let $\idx{\tau}{ADRE} \in
  \subnormstates{\hi{ADRE}}$ be a purification of $\idx{\tau}{ADR}$, then
  \begin{align*}
    \chmaxeps{\eps}{A}{DR}{\tau} &= - \chmineps{\eps}{A}{E}{\tau} \\
    &\leq -\chmineps{\eps}{A}{ER}{\tau} =
    \chmaxeps{\eps}{A}{D}{\tau}\, ,
  \end{align*}
  which concludes the proof.
\end{proof}

The second pair of data-processing inequalities concerns projective
(von Neumann) measurements of the system A. Such measurements can be
described in terms of an orthonormal basis $\{ \keti{i}{A} \}_i$ of
$\hA$ and a TP-CPM $\mathcal{M}$ from $\hA$ to $\hX
\iso \hA$ which maps $\rhoA$ to $\sum_i \bracket{i}{\rhoA}{i}\,
\proji{i}{i}{X}$. We expect that the uncertainty about the system A as
well as the entropies $\chmineps{\eps}{AB}{C}{}$ and
$\chmaxeps{\eps}{AB}{C}{}$ will not decrease with such a measurement.
\begin{theorem}
  \label{thm:data-proc2}
  Let $\eps \geq 0$, $\rhoABC \in \subnormstates{\hABC}$ and
  $\mathcal{M} : \posops{\hA} \to \posops{\hX}$ a TP-CPM
  describing a projective measurement with $\idx{\tau}{XBC} :=
  (\mathcal{M} \kron \opidi{BC})(\rhoABC)$. Then,
  \begin{align*}
    &\chmineps{\eps}{AB}{C}{\rho}
    \leq \chmineps{\eps}{XB}{C}{\tau} \quad \textrm{and} \\
    &\chmaxeps{\eps}{AB}{C}{\rho} \leq \chmaxeps{\eps}{XB}{C}{\tau} \,.
  \end{align*}
\end{theorem}
\begin{proof}
  Note that $\mathcal{M}$ can be
  decomposed into an isometry $U: \hA \to \hX \kron \hX'$, $\hX' \iso
  \hX$ that maps $\keti{i}{A}$ to $\keti{i}{X} \kron \keti{i}{X$'$}$
  followed by a partial trace over $\hX'$. We denote the intermediate
  state by $\idx{\tau}{XX$'$BC}$ and the projector onto the image of
  $U$ by $\idx{\Pi}{XX$'$}$. Moreover, Note that $\mathcal{M}(\idA) = \idX$.

  We first prove the statement for the min-entropy. Let $\rhotABC \in
  \epsball{\eps}{\rhoABC}$ and $\sigmaC \in \posops{\hC}$ such that
  $\chmineps{\eps}{AB}{C}{\rho} =
  \chmin{AB}{C}{\rhot} = - \log \tr\,\sigmaC$. Then, $\rhotABC \leq
  \idAB \kron \sigmaC$ implies
  \begin{align*}
   \big(\mathcal{M} \kron \opidi{BC} \big)(\rhotABC) \leq
    (\mathcal{M} \kron \opidi{BC}) (\idAB \kron \sigma)
    = \idi{XB} \kron \sigmaC \, .
  \end{align*}
  The state $\idx{\taut}{XBC} := (\mathcal{M} \kron \opidi{BC})(\rhotABC)$ is in
  $\epsball{\eps}{\idx{\tau}{XBC}}$ due to
  Lemma~\ref{lemma:purified-monotonicity}. Hence, $\idx{\taut}{XBC}$
  with $\sigmaC$ is a candidate for the optimization in
  $\chmineps{\eps}{XB}{C}{\tau}$ and, thus,
  $\chmineps{\eps}{XB}{C}{\tau} \geq \chmineps{\eps}{AB}{C}{\rho}$.

  To prove the statement for the max-entropy, we let $\idx{\taub}{XBC}
  \in \epsball{\eps}{\idx{\tau}{XBC}}$ be such that
  $\chmaxeps{\eps}{XB}{C}{\tau} = \chmax{XB}{C}{\taub}$. We
  use Corollary~\ref{cor:purified-ext} to introduce its extension
  $\idx{\taub}{XX$'$BC} \in \epsball{\eps}{\idx{\tau}{XX$'$BC}}$. Furthermore, we
  employ~\eqref{eqn:max-alt} to get
  \begin{align*}
    \chmaxeps{\eps}{XB}{C}{\tau} &= \!\! \max_{\sigmaB \in
      \normstates{\hB}} \log F \big(\idx{\taub}{XBC}, \,\idi{XB} \kron
    \sigmaC \big)^2 \\
    &\geq \!\! \max_{\sigmaB \in \normstates{\hB}} \log F
    \big(\idx{\taub}{XX$'$BC},\, \idx{\Pi}{XX$'$} \kron \idB \kron \sigmaC
    \big)^2 \\
    &=  \!\! \max_{\sigmaB \in \normstates{\hB}} \log F \big(
    \idx{\breve{\tau}}{XX$'$BC},\, \idi{XX$'$B} \kron \sigmaC \big)^2 \\
    &= \chmax{XX$'$B}{C}{\breve{\tau}} \, ,
  \end{align*}
  where we used that the fidelity can only increase under partial
  trace and introduced the state $\idx{\breve{\tau}}{XX$'$BC} :=
  (\idx{\Pi}{XX$'$} \kron \idBC) \idx{\taub}{XX$'$BC}
  (\idx{\Pi}{XX$'$} \kron \idBC)$. We have $\idx{\breve{\tau}}{XX$'$BC} \in
  \epsball{\eps}{\idx{\tau}{XX$'$BC}}$ due to the definition of
  $\idx{\taub}{XX$'$BC}$ and Lemma~\ref{lemma:purified-monotonicity}.
  We use this and Lemma~\ref{lemma:maxepsiso} to write
  $\chmax{XX$'$B}{C}{\breve{\tau}} \geq \chmaxeps{\eps}{XX$'$B}{C}{\tau} =
  \chmaxeps{\eps}{AB}{C}{\rho}$, from which the lemma follows.
\end{proof}

Note that, in conjunction with the fully quantum
generalization of the AEP (Theorem 1 in~\cite{tomamichel08}), the inequalities in
Theorem~\ref{thm:data-proc} and~\ref{thm:data-proc2} imply 
the same inequalities for the von Neumann entropy.

\appendix

\section{Technical Results}
\label{app:tech}

Here, we establish some useful properties of the min- and max-entropies.
In particular, we give bounds on the min- and max-entropies in terms of the
Hilbert space dimensions%
, show their
\comment{Lipschitz} continuity as a function of the state
and prove that the max-entropy is concave%
. Properties analogous to the ones we present here are also found for the
von Neumann entropy (see e.g.~\cite{nielsen00,alicki03}).

\subsection{Preliminaries}

Let us consider the functional $\Phi: \rhoAB \mapsto
2^{-\chmin{A}{B}{\rho}}$, which we extend to arbitrary Hermitian
operators, $\hermops{\hAB}$, on $\hAB$ as follows:
\begin{equation*}
\label{eqn:ext-guessing-prob}
\Phi : \hermops{\hAB} \to \mathbb{R} \ , \ \rhoAB\ \mapsto
\mathop{\inf_{\sigmaB \in \hermops{\hB}}}_{\rhoAB \leq \idA \kron \sigmaB}  \tr\ \sigmaB \, .
\end{equation*}
The functional has the following properties:
\begin{enumerate}

\item[i)] Multiplication with scalar:
  Let $\lambda \geq 0$, then
  $\Phi(\lambda \rhoAB) = \lambda\,\Phi(\rhoAB)$ .

\item[ii)]
  Monotonicity: \
  $\rhoAB \geq \tauAB \implies \Phi(\rhoAB) \geq
  \Phi(\tauAB)$.

\item[iii)] Sub-Additivity:
  $\Phi(\rhoAB + \tauAB)
  \leq \Phi(\rhoAB) +
  \Phi(\tauAB)$. Furthermore, equality holds if $\trace{\rhoB \tauB} = 0$.

\item[iv)] Bounds:
  Let $\idx{d}{A} = \dim \hA$ and $d_\textrm{min} = \min \{ \idx{d}{A},
  \dim \hB \}$, then
  $\frac{1}{\idx{d}{A}} \tr\,\rhoAB  \leq
  \Phi(\rhoAB) \leq d_\textrm{min}\,\tr\,\{\rhoAB\}_+$.
  \begin{proof}
    To get the upper bound, first consider a normalized pure state
    $\phiAB$. Clearly, $\phiAB \leq \PiA^\phi \kron \PiB^\phi \leq
    \idA \kron \PiB^\phi$, where $\PiA^\phi$ and $\PiB^\phi$ are the
    projectors onto the support of $\phiA$ and $\phiB$, respectively.
    Furthermore $\tr\, \PiB^\phi \leq d$ thanks to the Schmidt
    decomposition. Using the eigenvalue decomposition $\rhoAB =
    \sum_i \lambda_i\,\phiAB^i$, we get
    \begin{equation}
      \label{eqn:guessing-upper-bound}
      \Phi(\rhoAB) \leq \tr\,\bigg( \sum_{i: \lambda_i > 0} \lambda_i\,
        \PiB^{\phi^i} \bigg) \leq d_\textrm{min}\, \tr\,\{\rhoAB \}_+\, .
    \end{equation}

    On the other hand, we have $\trace{\idA \kron \sigmaB} \geq
    \tr\,\rhoAB$ for any candidate $\sigmaB$, hence,
    \begin{equation}
      \label{eqn:guessing-lower-bound}
      \Phi(\rhoAB) \geq \frac{1}{\idx{d}{A}} \tr\,\rhoAB \, .
    \end{equation}
  \end{proof}

\end{enumerate}

Properties i) and iii) imply convexity of $\Phi$, i.e.\ $\Phi(\lambda \rhoAB + (1 - \lambda) \tauAB) \leq \lambda \Phi(\rhoAB) + (1-\lambda) \Phi(\tauAB)$.

\subsection{Bounds on the Conditional Entropies}
\label{app:bounds}

In \cite{tomamichel08} it was shown that, for $\rhobAB \in
\normstates{\hAB}$, we have $\chmin{A}{B}{\rhob} \leq
\chmax{A}{B}{\rhob}$. For sub-normalized states $\rhoAB = \tr\,\rhoAB
\cdot \rhobAB$, we thus have
\begin{equation}
  \label{eqn:ordering-subnorm}
  \chmin{A}{B}{\rho} + \log \tr\,\rhoAB
  \leq \chmax{A}{B}{\rho} - \log \tr\,\rhoAB \, .
\end{equation}

We now establish bounds on the min- and max-entropies:
\begin{lemma}
Let $\rhoAB \in \subnormstates{\hAB}$, $\idx{d}{A} = \dim \hA$ and $d_\textnormal{min} = \min\{ \idx{d}{A}, \dim \hB\}$, then
\begin{align}
\label{eqn:min-bounds}
- \log d_\textnormal{min} \ \leq\ & \chmin{A}{B}{\rho} + \log \tr\,\rhoAB\ \leq\ \log \idx{d}{A} \,  \nonumber\\
- \log d_\textnormal{min} \ \leq\ & \chmax{A}{B}{\rho} - \log \tr\,\rhoAB\ \leq\ \log \idx{d}{A} \, . \nonumber
\end{align}
\end{lemma}

\begin{proof}
  The bounds on the min-entropy follow directly
  from~\eqref{eqn:guessing-upper-bound}
  and~\eqref{eqn:guessing-lower-bound}. The bounds on the max-entropy
  follow by duality~\eqref{eqn:duality}
  and~\eqref{eqn:ordering-subnorm}.
\end{proof}

\subsection{Continuity of the Conditional Entropies}
\label{app:cont}

The operational interpretation of the conditional min-entropy as a
guessing probability (cf.\ \cite{koenig08}) already implies its
continuity in the state. To see this, note that a discontinuity in
the guessing probability could be detected experimentally using a
fixed number of trials (the number depending only on the required
precision), hence giving us the means to distinguish between
arbitrarily close states for a cost (in terms of the number of
trials) independent of their distance. For sufficiently close
states, this would contradict the upper bound on the distinguishing
advantage~\eqref{eqn:dist-adv}. Here, we make this statement more
precise.
\begin{lemma}
\label{lemma:min-cont} Let $\rhoAB, \tauAB \in \subnormstates{\hAB}$
and $\delta := \dist(\rhoAB, \tauAB)$, then
\begin{equation*}
\big| \chmin{A}{B}{\rho} - \chmin{A}{B}{\tau} \big| \leq \frac{
\idx{d}{A} d_\textnormal{min}\, \delta }{ \ln 2 \cdot \min\{
  \tr\,\rhoAB,\, \tr\,\tauAB \} } \, .
\end{equation*}
\end{lemma}
\begin{proof}
  We use continuity of the functional $\Phi$ to obtain
  \begin{align*}
    \Phi(\tauAB) &= \Phi(\rhoAB + (\tauAB - \rhoAB)) \leq
    \Phi(\rhoAB) + \Phi(\tauAB - \rhoAB) \\
    & \leq \Phi(\rhoAB) + d_\textnormal{min}\,\tr\{\tauAB -
    \rhoAB \}_+ \leq \Phi(\rhoAB) + d_\textnormal{min}\, \delta\, .
  \end{align*}
  Note that $\Phi > 0$ for all states in $\subnormstates{\hAB}$.
  Taking the logarithm and using the bound $\ln (a + x) \leq \ln a + \frac{x}{a}$, we find
  \begin{align*}
    \log \Phi(\tauAB) - \log \Phi(\rhoAB) \leq
    \frac{d_\textnormal{min}\,\delta}{\ln 2 \cdot \Phi(\rhoAB)} \leq \frac{\idx{d}{A}
      d_\textnormal{min}\, \delta}{\ln 2 \cdot \tr\,\rhoAB} \, .
  \end{align*}
  The same argument also applies on exchange of $\rhoAB$ and $\tauAB$
  and we obtain the statement of the lemma by substituting $\chmin{A}{B}{\rho} = - \log
  \Phi(\rhoAB)$.
\end{proof}

\begin{remark}
  The above result is tight in the following sense: Consider a system
  with Hilbert spaces $\hA$ and $\hB = \hA' \oplus \hB'$, where $\hA'
  \iso \hA$. Let $\psiAB$ be the normalized fully entangled state on
  $\hA \kron \hA'$ and $\rhoB \in \subnormstates{\hB'}$ be orthogonal
  to $\psiB$. The choice $\rhoAB = \frac{\idA}{\idx{d}{A}} \kron
  \rhoB$ and $\tauAB = \rhoAB + \delta\, \psiAB$ for some small
  $\delta > 0$ leads to $\dist(\rhoAB, \tauAB) = \delta$,
  $$\Phi(\rhoAB) = \frac{\tr\,\rhoB}{\idx{d}{A}} \quad \textrm{and}
  \quad \Phi(\tauAB) =
  \Phi(\rhoAB) + d_\textnormal{min}\, \delta\, .$$
  Taking the logarithm (for small $\delta$) leads to
  $$\log \Phi(\tauAB) - \log \Phi(\rhoAB) \approx \frac{
    d_\textnormal{min}\, \delta}{\ln 2 \cdot \Phi(\rhoAB)} = \frac{\idx{d}{A}
    d_\textnormal{min}\,\delta }{ \ln 2 \cdot \tr\,\rhoAB} \, .$$
\end{remark}

Lemma~\ref{lemma:min-cont} implies that the conditional min-entropy
is uniformly (Lipschitz) continuous on the set of normalized states
and in any $\eps$-ball. Since $\bar{D}(\rho, \tau) \leq P(\rho,
\tau)$ (cf.\ Lemma~\ref{lemma:metric-bounds}),
Lemma~\ref{lemma:min-cont} also holds for $\delta = P(\rhoAB,
\tauAB)$.

The continuity of the smooth min- and max-entropies then follows:
Let $\rhotAB \in \epsball{\eps}{\rhoAB}$ be such that
$\chmineps{\eps}{A}{B}{\rho} = \chmin{A}{B}{\rhot}$. We now
construct a state $\tautAB$ that is $\eps$-close to $\tauAB$ and
$\delta'$-close to $\rhotAB$, where $\delta' := \sqrt{\delta^2 +
  2\eps\delta}$.\footnote{The construction is as follows: Let $c :=
  (\delta + \eps)^2$, $\varphi$ be a purification of
  $\rhotAB \oplus (1 - \tr\,\rhotAB)$ and $\vartheta$ a be purification
  of $\tauAB \oplus (1 - \tr\,\tauAB)$ such that
  $\abs{\braket{\varphi}{\vartheta}}^2 = \bar{F}(\rhotAB,
  \tauAB)^2 \geq 1 - c$. We choose $\tautAB := c^{-1} (\eps^2
  \rhotAB + \delta'^2 \tauAB)$.
  Now, $P(\tautAB, \tauAB)^2 \leq P \big(
  c^{-1}(\eps^2 \varphi + \delta'^2 \vartheta), \vartheta \big)^2 = 1 -
  c^{-1} \bracket{\vartheta}{\eps^2 \varphi + \delta'^2
    \vartheta}{\vartheta} \leq 1 - c^{-1} \big(\delta'^2 + \eps^2 ( 1 -
  c) \big) = \eps^2$. Similarly, $P(\tautAB,
  \rhotAB) \leq \delta'$.}
We get
\begin{align*}
  \chmineps{\eps}{A}{B}{\rho} - \chmineps{\eps}{A}{B}{\tau} \leq
  \chmin{A}{B}{\rhot} - \chmin{A}{B}{\taut} \, ,
\end{align*}
which vanishes continuously for $\delta \to 0$ due to
Lemma~\ref{lemma:min-cont}.
The continuity of the smooth max-entropy follows by
duality~\eqref{eqn:duality}. Using
Lemma~\ref{lemma:purified-monotonicity}, we introduce purifications
$\rhoABC$ of $\rhoAB$ and $\tauABC$ of $\tauAB$ such that $P(\rhoAB,
\tauAB) = P(\rhoABC, \tauABC) \geq P(\rhoAC, \tauAC)$. Then,
$$\big| \chmaxeps{\eps}{A}{B}{\rho} - \chmaxeps{\eps}{A}{B}{\tau}
\big| = \big| \chmineps{\eps}{A}{C}{\rho} -
\chmineps{\eps}{A}{C}{\tau} \big|\, $$ which can be bounded using
Lemma~\ref{lemma:min-cont} with $d_\textrm{min} = \idx{d}{A}$.

\subsection{Concavity of the Max-Entropy}
\label{app:concave}

The max-entropy is a concave function of the state.

\begin{lemma}
  Let $\{p_i\}_i$ be a probability distribution, $\{ \rhoAB^i \}_i$ be
  a set of states in
  $\subnormstates{\hAB}$ and $\tauAB := \sum_i p_i\, \rhoAB^i$. Then,
  $$ \chmax{A}{B}{\tau} \geq \sum_i p_i\, \chmax{A}{B}{\rho^i} \, . $$
\end{lemma}

\begin{proof}
  Let $\keti{\varphi^i}{ABC}$ purify $\rhoAB^i$ such that the state
  $\keti{\tau}{ABCYZ} := \sum_i \sqrt{p_i}\,
  \keti{\varphi^i}{ABC} \kron \keti{i}{Y} \kron \keti{i}{Z}$ --- where
  $\{ \keti{i}{Y} \}_i$ and $\{ \keti{i}{Z} \}_i$ are orthonormal
  bases of the auxiliary Hilbert spaces $\hY$ and $\hi{Z}$,
  respectively --- has marginals $\tauAB$ and
  $\idx{\tau}{ACZ} = \sum_i p_i\,\rhoAC^i \kron \proji{i}{i}{Z}$.
  Using data-processing of the max-entropy
  (Theorem~\ref{thm:data-proc}), the properties of $\Phi$ and the
  concavity of the logarithm, we find
  \begin{align}
    \chmax{A}{B}{\tau} & \geq \chmax{A}{BY}{\tau} =  -
    \chmin{A}{CZ}{\tau} \nonumber\\
    & = \log \Phi (\idx{\tau}{A(CZ)}) = \log \Big( \sum_i p_i\,
      \Phi(\idx{\rho}{AC}^i) \Big) \nonumber\\
    & \geq \sum_i p_i \log \Phi(\idx{\rho}{AC}^i) = \sum_i p_i \,
    \chmax{A}{B}{\rho^i} \, . \nonumber
  \end{align}
\end{proof}

Note also that the min-entropy is neither a concave nor a convex function
of the state.

\section*{Acknowledgment}

We thank Nilanjana Datta, J\"urg Wullschleger and Christian Schaffner for fruitful discussions and comments.
We acknowledge support from the Swiss National Science Foundation (grant No. 200021-119868).


\bibliographystyle{apsrev}

\end{document}